\RequirePackage{fix-cm}
 \documentclass[smallextended]{svjour3}       
\smartqed  
\usepackage{graphicx}
\usepackage{array}
\usepackage{amsmath}
\newtheorem{Theorem}{Theorem}
\newtheorem{Lemma}{Lemma}

\begin{document}

\title{ Quantum state representation based on combinatorial Laplacian matrix of star-relevant graph
}


\author{\small Jian-Qiang Li \and Xiu-Bo Chen \and Yi-Xian Yang}


\institute{Jian-Qiang Li \and Xiu-Bo Chen \and Yi-Xian Yang\at
1)Information Security Center, State Key Laboratory of Networking and Switching Technology, Beijing University of Posts and Telecommunications, Beijing 100876, China\\
\\
 \email{ljqb206@163.com}  \email{flyover100@163.com}  (Corresponding author Xiubo Chen)
}

\date{Received: date / Accepted: date}

\maketitle

\begin{abstract}

In this paper the density matrices derived from combinatorial laplacian matrix of graphs is considered. More specifically, the paper places emphasis on the star-relevant graph, which means adding certain edges on peripheral vertices of star graph. Initially, we provide the spectrum of the density matrices corresponding to star-like graph(i.e., adding an edge on star graph) and present that the Von Neumann entropy increases under the graph operation(adding an edge on star graph) and the graph operation cannot be simulated by local operation and classical communication (LOCC). Subsequently, we illustrate the spectrum of density matrices corresponding to star-alike graph(i.e, adding one edge on star-like graph) and exhibit that the Von Neumann entropy increases under the graph operation(adding an edge on star-like graph)and the graph operation cannot be simulated by LOCC. Finally, the spectrum of density matrices corresponding to star-mlike graph(i.e., adding $m$ nonadjacent edges on the peripheral vertices of star graph) is demonstrated and the relation between the graph operation and Von Neumann entropy, LOCC is revealed in this paper.

%
\end{abstract}

\section{Introduction}
\qquad 

 During the last century, different ideas have been developed to describe quantum phenomena. Hesienberg, Jordan, and Dirac brought up ``matrix mechanics" in 1925, later Schrodinger proposed the `` wave mechanics " from an entirely different starting point and later proved it equivalent to the former\cite{1}. All these efforts were made to describe the physical problem using simple mathematics. And quantum state and observable are the two main ingredients necessary to describe physical phenomena. Taking advantage of these advancements, various mathematical concepts have been proposed as quantum state representation, such as state vector, density operator\cite{2}, linear functional\cite{3} and graph\cite{4}. Until now, loads of fascinating results are provided based on these quantum state representations, for instance, Shor's quantum factoring algorithm\cite{5}, quantum teleportation\cite{6,7}, quantum dense coding\cite{8,9} and quantum cryptography protocols\cite{31,32,33} etc. Nevertheless, there are still many unknown properties waiting to be explored.

   Graph theory, as a well-developed mathematical theory\cite{10,11}, is highly applicable in the diverse areas such as network system\cite{12} and optimization\cite{13}. Thus the integration of quantum state representation and graph theory can offer some attractive features of quantum states. Generally, there are two approaches of mapping graphs into quantum states. First way is to use the vertex and edge to represent quantum state and the interaction between quantum state, respectively. A good example here is graph-state, which have been widely applied in stabilizer states, cluster states and some multipartite entangled states\cite{14}. The other way is based on a faithful mapping between discrete Laplacians and quantum states (while the converse is not necessary), which was firstly introduced by Braunstein \textit{et al}.\cite{15}. Later, by considering graphs with complex edge weights and loops with real positive weights, Adhikari \textit{et al}.\cite{17} provided a construction of graphs corresponding to general quantum states and establish a one-to-one mapping between graphs and quantum states.

  With the latter quantum representation, several aspects of quantum state representation based on graph, including separability,Von Neumann entropy and quantum operation, has been investigated. The simple combinatorial condition(degree condition) characterizes the separability of density matrices of graphs\cite{26,27,28,29,30}. Filippo \textit{et al}.\cite{18} analyzed the entropy quantity for various graphs, such as regular graphs, random graphs and put forward an open problem ``Does the star graph has smallest entropy among all connected graphs on $n$ vertices?'' Braunstein \textit{et al}.\cite{15} proposed the quantum operation needed to implement the graph operation, such as deleting or adding an edge, deleting or adding a vertex. And the LOCC operation increases the Von Neumann entropy\cite{21}.

     In this paper, by presenting the spectrum of the density operator corresponding to star-relevant graphs, which means adding certain edges on star graphs, the faithful mapping between star-relevant graphs and quantum states is  established. Also taking advantage of spectrum of the star-relevant graph, we demonstrate that the Von Neumann entropy is increased by adding an edge on star-relevant graph and the graph operation cannot be simulated by LOCC.

    The rest of the paper is organized as follows: Section II introduces preliminaries knowledge about the quantum state representation with graph; Section III presents the spectrum of density operator obtained from combinatorial laplacian matrix of the star-like graph, which means adding an edge on star graph. Also the relation between  the graph operation adding an edge on star-graph and Von Neumann entropy, LOCC is illustrated. Section IV presents the spectrums corresponding to  two types of star-alike graph, which means adding an edge on star-like graph. Meanwhile the relation between the graph operation adding an edge on star-like graph and Von neumann entropy, LOCC is demonstrated. Furthermore, the relationship of Von Neumann entropy, LOCC  between $star-alike_1$ and $star-alike_2$ graph is also investigated. Section V indicates spectrum corresponding to the star-mlike graph, which means adding $m$ nonadjacent edges on peripheral vertices of star graph. Analogously, the relation between the graph operation adding an non-adjacent edge on the peripheral vertices of star-relevant graph and Von Neumann entropy, LOCC is also displayed. Finally, we discuss the weighted graph situation and make an conclusion.

\section{Preliminaries}
\qquad \label{sec:1}
   Quantum state representation is associated with the density operator, which is a matrix in finite dimensional quantum system. And graph theory has intimated the relationship with matrix at the same time, thus quantum state representation can be built from graphs. This section introduces the basic concepts of graph, density operator and the way mapping graphs into quantum states. In order to investigate the relationship between graph operation and the feature of quantum states, the concepts about Von Neumann entropy and LOCC are introduced.

\subsection{Graph}
\qquad \label{sec:2}
 Mathematically, a graph G is a pair of sets $(V,E)$ where $V$ is a finite nonempty set containing the vertices and E is a set of edges that connect pair of vertices. The number of edges adjacent to a vertex is called its degree.

    Let $G=(V,E)$ be a simple undirected graph with set of vertices $V(G)={1,2,\cdots, n}$
    and  $E(G)\subseteq V(G)\times V(G)- \{(v,v): v\in V(G) \}$.

    The adjacency matrix associated with $G$ is denoted by $A(G)$ and defined by $[A(G)]_{uv}=1$ if $(u,v)\in E(G)$ and $[A(G)]_{uv}=0$ otherwise.

    The degree of a vertex $v\in V(G)$, denoted by $d(v)$ is the number of edges adjacent to $v$, the degree matrix of $G$ is an $n\times n$ matrix, denoted by $\bigtriangleup(G)$ and defined by $[\bigtriangleup(G)]_{u,v}=d(v)$ if $u=v$ and  $[\bigtriangleup(G)]_{u,v}=0$, otherwise.\cite{4}
 \subsection{Density operator}
\label{sec:3}
\qquad In quantum mechanics, a quantum state is described by a density operator $\rho$ in $n$ dimensions Hilbert space, which associates with a quantum system.
Actually, two conditions are necessary and sufficient to describe a statistical ensemble with density operator.
   \begin{enumerate}
    \item  normalization  $tr(\rho)=1$.
    \item $\rho$ is a positive operator $\rho \geq 0$
   \end{enumerate}
 This is helpful for deriving density operator from graph.
\subsection{Graph and Density operator}
\qquad \label{sec:4}
 The density operator can also be built from graphs in theory, the explicit construction is explained as follows:

Initially, the combinatorial Laplacian matrix of a graph $L(G)$ is the matrix
    \begin{equation}\label{*}
        L(G)=\bigtriangleup(G)-A(G),
    \end{equation}\\
where $L(G)$ is positive semi-definite according to the Gersgorin disc theorem\cite{20} and then scaling the Laplacian of a graph $G$  by the $d_{G}$,
\begin{equation}\label{*}
    \rho_{G}=\frac{L(G)}{d_{G}}=\frac{L(G)}{Tr(\triangle(G))},
\end{equation}
here $d_{G}=\sum_{v\in V(G)}d(v)$.

   $\rho_G$ is a density operator deduced from the fact that $tr(\rho_G)=1$ and $\rho_G$ is a positive operator $\rho \geq 0$.
Thus, with proper construction, the faithful mapping between quantum state and graph is established.
\subsection{Von Neumann entropy}
\label{sec:5}
 \qquad Entropy is an important way to measure the uncertainty associated with a classical probability.  When it comes to quantum world,
   Von Neumann defined the entropy of a quantum state $\rho$ by the formula :
       \begin{equation}\label{1}
        S(\rho)=-tr(\rho \log\rho)
       \end{equation}\\
   Here, logarithm is taken with respect to base two and the quantum state is described by density operator $\rho$.

  If $\lambda_{x}$ are the eigenvalues of $\rho$ then the Von Neumann's definition can be re-expressed as:
      \begin{equation}\label{*}
        S(\rho)=-\sum_x\lambda_{x}\log\lambda_{x}
      \end{equation}
  Thus, we have Von Neumann entropy to measure the uncertainty of a quantum state\cite{2}.
  \subsection{LOCC}
\label{sec:6}
\qquad In quantum communication, LOCC is a way to reduce the cost of quantum resources and plays an important role in many quantum information tasks since quantum resources are more valuable. It has been demonstrated that LOCC has a close relationship with majorization\cite{24}, which is an important mathematical method in optimization fields\cite{25}. That relationship is further clarified with Nielsen's theorem.

  We take the sight of majorization first, suppose $x=(x_1,x_2,\cdots,x_n)$, $y=(y_1,y_2,\cdots,y_n)$ and use notation $x^\downarrow$ to denote the components of $x$
  in decreasing order, for instance, $x_1^\downarrow$ means the largest component of $x$. We say $x$ is majorized by $y$,  written
  $x\prec y$, if $\sum_{j=1}^{k}x_j^\downarrow \leq \sum_{j=1}^{k}y_j^\downarrow$, here $k=1,\cdots,n$, with equality when $k=n$.

  Now we focus on  Nielsen's theorem, suppose $|\psi\rangle$ and $|\phi\rangle$ are pure entangled states shared by Alice and Bob.
  Define $\rho_\psi\equiv tr_B(|\psi\rangle\langle\psi|), \rho_\phi\equiv tr_B(|\phi\rangle\langle\phi|)$ to be the corresponding
  reduced density matrices of Alice's system, and let $\lambda_\psi$ and $\lambda_\phi$ be the vectors whose entries are the eigenvalues
  of $\rho_\psi$ and $\rho_\phi$.

 \textbf{Nielsen's theorem}
        A bipartite pure state $|\psi\rangle$ may be transformed into another pure state $|\phi\rangle$
        by LOCC if and only if $\lambda_{\psi}\prec \lambda_{\phi}$\cite{2}.

  Therefore, the criteria to transform a quantum state to another one with LOCC is often reduced to the majorization condition. Also the relation between Von Neumann entropy and the majorization condition is illustrate in \cite{21}.
  \begin{Lemma}
      If $\lambda_{\psi}\prec \lambda_{\phi}$, It holds that  $S(\rho_\psi)\geq S(\rho_\phi)$.
\end{Lemma}
\section{Star-like graph, adding an edge on star graph}
\label{sec:7}
\qquad

\begin{figure}
  \includegraphics[height=4cm,width=15cm]{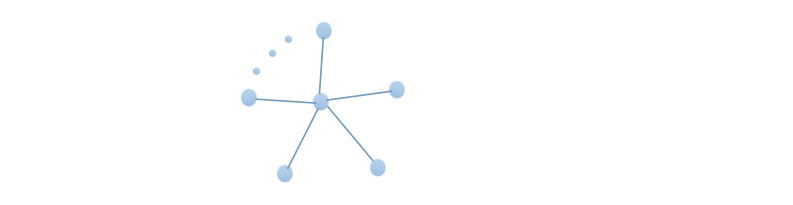}
 \caption{Star graph}
\label{fig:1}       
\end{figure}

 In this section the specific quantum state corresponding to star-like graph, which means adding one edge on star graphs(Fig.1), is introduced. Also the relationship of Von Neumann entropy between star graph and star-like graph(Fig.2) is investigated.
\begin{figure}
  \includegraphics[height=4cm,width=15cm]{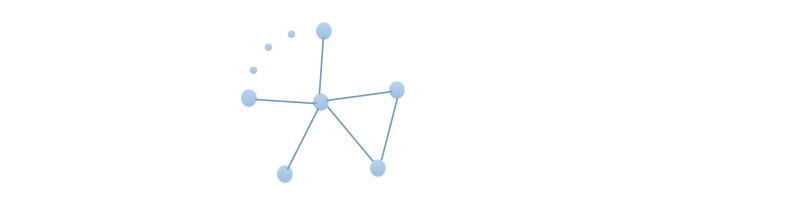}
 \caption{Star-like graph}
\label{fig:1}       
\end{figure}

 \subsection{Spectrum of the density operator corresponding to star-like graph}
 \label{sec:8}
  According to the method in section 2.3, for star-like graphs, the adjacent and degree matrix are presented as
\begin{eqnarray}
A(G)=
\left(
\begin{array}{cccccc}
0&1 & \ldots&1 &1&
\\
1& 0&\ldots & 0 &0&
\\
\vdots &\vdots & 0  &\vdots &\vdots
\\
1&0 &\ldots& 0& 1
\\
1&0 &\dots &1& 0
\end{array}
\right)
\end{eqnarray}
and
  \begin{eqnarray}
\triangle(G)=
\left(
\begin{array}{cccccc}
n-1&0 & \ldots&0 &0&
\\
0& 1&\ldots & 0 &0&
\\
\vdots &\vdots & 1  &\vdots &\vdots
\\
0&0 &\ldots& 2 & 0
\\
0&0 &\dots &0& 2
\end{array}
\right)
\end{eqnarray}\\
%
respectively.

  Hence the laplcian matrix of star-like graphs is given as follows:
        \begin{center}$L(G)=\bigtriangleup(G)-A(G)$\end{center}
  The density operator corresponding to star-like graph as follows\\
 \begin{eqnarray}
\rho(G)=\frac{L(G)}{d_{G}}=
\left(
\begin{array}{cccccc}
\frac{n-1}{2n}& \frac{-1}{2n} & \ldots&\frac{-1}{2n} &\frac{-1}{2n}&
\\
\frac{-1}{2n}& \frac{1}{2n}&\ldots & 0 &0&
\\
\vdots &\vdots & \frac{1}{2n}  &\vdots &\vdots
\\
\frac{-1}{2n}&0 &\ldots&\frac{2}{2n}& \frac{-1}{2n}
\\
\frac{-1}{2n}&0 &\dots &\frac{-1}{2n} &\frac{2}{2n}
\end{array}
\right)
\end{eqnarray}
where  $d_{G}=\sum_{v\in V(G)}d(v)=2n$.

In order to confirm the specific quantum states, it is necessary to show the spectrum of this density operator.

\begin{Theorem}
The spectrum of density operator corresponding to star-like graphs is $\{{\frac{1}{2}^{[1]}, \frac{3}{2n}^{[1]}, \frac{1}{2n}^{[n-3]}, 0^{[1]}}\}$.\\
 \end{Theorem}
\begin{proof}
:The finding of the spectrum is equivalent to solve the equation\\

 $det(\lambda I-\rho)=0$\\
 that is\\

\begin{eqnarray}
det(\lambda I-\rho)=
\left|
\begin{array}{cccccc}
\lambda-\frac{n-1}{2n}& \frac{1}{2n} & \ldots&\frac{1}{2n} &\frac{1}{2n}&
\\
\frac{1}{2n}& \lambda-\frac{1}{2n}&\ldots & 0 &0&
\\
\vdots &\vdots & \lambda-\frac{1}{2n}  &\vdots &\vdots
\\
\frac{1}{2n}&0 &\ldots& \lambda- \frac{2}{2n}& \frac{1}{2n}
  \\
\frac{1}{2n}&0 &\dots &\frac{1}{2n} & \lambda-\frac{2}{2n}
\end{array}
\right|=0
\end{eqnarray}\\
 $=(\lambda-\frac{n-1}{2n})(\lambda-\frac{1}{2n})^{n-3}[(\lambda-\frac{2}{2n})^{2}-(\frac{1}{2n})^{2}]$\\
 $-(n-3)(\frac{1}{2n})^{2}(\lambda-\frac{1}{2n})^{n-4}[(\lambda-\frac{2}{2n})^{2}-(\frac{1}{2n})^{2}]$\\
 $-2(\frac{1}{2n})^{2}(\lambda-\frac{1}{2n})(\lambda-\frac{3}{2n})(\lambda-\frac{1}{2n})^{n-4}$\\
 $=(\lambda-\frac{n-1}{2n})(\lambda-\frac{1}{2n})^{n-3}(\lambda-\frac{3}{2n})$\\
 $-(n-3)(\frac{1}{2n})^{2}(\lambda-\frac{1}{2n})^{n-3}(\lambda-\frac{3}{2n})$\\
 $-2(\frac{1}{2n})^{2}(\lambda-\frac{1}{2n})(\lambda-\frac{3}{2n})(\lambda-\frac{1}{2n})^{n-4}$\\
 $=(\lambda-\frac{1}{2n})^{n-3}(\lambda- \frac{3}{2n})[(\lambda-\frac{n-1}{2n})(\lambda-\frac{1}{2n})-(n-3)\frac{1}{(2n)^{2}}]$\\
 $=(\lambda-\frac{1}{2n})^{n-3}(\lambda- \frac{3}{2n})(\lambda^{2}-\frac{1}{2}\lambda)$\\
 $=\lambda(\lambda-\frac{1}{2})(\lambda-\frac{3}{2n})(\lambda-\frac{1}{2n})^{n-3}$\\
 $=0$\\

 Clearly,  the roots are $\{{\frac{1}{2}^{[1]}, \frac{3}{2n}^{[1]}, \frac{1}{2n}^{[n-3]}, 0^{[1]}}\}$, which implies that the spectrum of density operator corresponding to star-like graphs is $\{{\frac{1}{2}^{[1]}, \frac{3}{2n}^{[1]}, \frac{1}{2n}^{[n-3]}, 0^{[1]}}\}$.  \qed
 \end{proof}

 The faithful (and exact) mapping and the specific quantum states are now connected with star-like graph. Also the preparation work for investigation the relationship between star-like graphs and quantum states is ready to deploy.

\subsection{Graph operation and Von Neumann entropy, LOCC}
\qquad \label{sec:9}

The aforementioned result presents only the quantum states representation survey and emphasizes on the specific quantum states corresponding to star-like graphs. A systematic investigation about the relation of graph operation and feature of quantum states remains to be considered. The graph operation is a mapping that transforms one graph to another, which can be implemented by quantum operation\cite{15}.  In this section, we limit ourselves to the proof of the relation between a particular graph operation and Von Neumann entropy, LOCC.
\subsubsection{ Graph operation and Von Neumann entropy}
\label{sec:10}
\qquad
The graph operation is an addition of edge to star graph, which transforms star graph into star-like graph, can be implemented by quantum operation\cite{15}. When comeing to the relationship with Von Neumann entropy, then

  \begin{Theorem} The corresponding Von Neumann entropy increases under the graph operation adding an edge on star graph.\\
   \end{Theorem}
   \begin{proof}
   :We have already known that $\{{\frac{n}{2n-2}^{[1]},  \frac{1}{2n-2}^{[n-2]}, 0^{[1]}}\}$ is the spectrum of density operator corresponding to star graphs\cite{11,18} and the spectrum of density operator corresponding to star-like graphs is $\{{\frac{1}{2}^{[1]}, \frac{3}{2n}^{[1]}, \frac{1}{2n}^{[n-3]}, 0^{[1]}}\}$.
    According to the definition of Von Neumann entropy, \\
$S(\rho_{star})=-\sum_x\lambda_{x}\log\lambda_{x} $\\
                    $=-\frac{n}{2n-2}\log(\frac{n}{2n-2})-(n-2)\frac{1}{2n-2}\log(\frac{1}{2n-2})$\\
                    $=\frac{n}{2n-2}\log(\frac{2n-2}{n})+(n-2)\frac{1}{2n-2}\log(2n-2)$\\
                    $=\frac{1}{2}\log(2-\frac{2}{n})+\frac{1}{2}\log(2n-2)-\frac{1}{2n-2}\log(n)$\\
    \\ the entropy corresponding to quantum states represented by star-like graphs is given as follows\\
     $ S(\rho_{star-like})=-\sum_x\lambda_{x}\log\lambda_{x}$\\
                    $ =-\frac{1}{2}\log(\frac{1}{2})-\frac{3}{2n}\log(\frac{3}{2n})-(n-3)\frac{1}{2n}\log(\frac{1}{2n})$\\
                    $=\frac{1}{2}\log2+\frac{3}{2n}\log(\frac{2n}{3})+\frac{n-3}{2n}\log(2n)$\\
                    $=\frac{1}{2}\log2+\frac{3}{2n}\log(\frac{1}{3})+\frac{1}{2}\log(2n)$\\
      \\ Then,\\
     $ S(\rho_{star-like})-S(\rho_{star})=$\\
     $[\frac{1}{2}\log2-\frac{1}{2}\log(2-\frac{2}{n})]$
                    $+[\frac{1}{2}\log(2n)-\frac{1}{2}\log(2n-2)]$
                    $+[(-\frac{3}{2n}\log3+\frac{1}{2n-2}\log n)]$\\

    $ (S(\rho_{star-like})-S(\rho_{star}))'$\\
    $=[-\frac{1}{2n-2} +\frac{1}{2n}]+[\frac{1}{2n}-\frac{1}{2n-2}]+[\frac{3}{2}\frac{1}{n^2}\log3+\frac{1}{(2n-2)n}-\frac{1}{2(n-1)^2}\log(n)]$\\
    $ =\frac{1}{n}-\frac{1}{n-1}+\frac{3}{2}\frac{1}{n^2}\log3+\frac{1}{(2n-2)n}-\frac{1}{2(n-1)^2}\log(n)$\\
    $=\frac{3}{2}\frac{1}{n^2}\log3-\frac{1}{(2n-2)n}-\frac{1}{2(n-1)^2}\log(n)$\\
    $=\frac{3\log3}{2 n^2}-\frac{n\log(n)+n-1}{2(n-1)^2 n}  $\\
    $ (S(\rho_{star-like})-S(\rho_{star}))'<0$ when $n\geq 3$. Hence, the entropy difference is decreasing with each increment in $n$. Also,
     $\lim_{n\rightarrow \infty} [S(\rho_{star-like})-S(\rho_{star})]=0$ and when $n=3$, $[S(\rho_{star-like})-S(\rho_{star})]>0$,
     hence the entropy difference always remains positive, which implies that  Von Neumann entropy increases under the graph operation by adding an edge on star graph. \qed
\end{proof}

The relationship between Von Neumann entropy and this graph operation is revealed in this section. As a byproduct, the theorem also partially answers the open problem put forward in\cite{18} `` Does the star graph has the smallest entropy among all connected graphs on $n$ vertices?'' And then, ``what is the internal relationship between graph operation and quantum operation?", more specifically, we want to know the relationship between this graph operation and LOCC.

\subsubsection{Graph operation and LOCC}
\label{sec:11}
\qquad A systematic investigation on the relationship between graph operation and LOCC goes beyond the scope of this section. Instead we focus on the relation of the graph operation, adding an edge on star graph and LOCC. Concretely, we try to answer the question ``Whether this graph operation can be simulated by LOCC?"

\begin{Theorem} Suppose $|\psi\rangle$ and $|\phi\rangle$ are pure entangled states shared by Alice and Bob. Let $\rho_{star}=tr_B(|\psi\rangle\langle\psi|)$,
    $\rho_{star-like}=tr_B(|\phi\rangle\langle\phi|)$ corresponding to Alice's system, The quantum states  cannot be transformed into states represented by star-like
    graph by the way of LOCC when $n\geq 4$.\\
\end{Theorem}
\begin{proof}
:The density operator and spectrum of the density operator corresponding to star and star-like graphs are given as follows:

   $\rho_{star}=tr_B(|\psi\rangle\langle\psi|)$ and
   $\lambda_{\psi}=\{{\frac{n}{2n-2}^{[1]}, \frac{1}{2n-2}^{[n-2]}, 0^{[1]}}\}$

 $\rho_{star-like}=tr_B(|\phi\rangle\langle\phi|)$  and
 $\lambda_{\phi}=\{{\frac{1}{2}^{[1]}, \frac{3}{2n}^{[1]}, \frac{1}{2n}^{[n-3]}, 0^{[1]}}\}$.\\
   With Nielsen theorem, the issue to decide the capability of $|\psi\rangle$ being transformed to $|\phi\rangle$ by LOCC is
   equivalent to determine whether $\lambda_\psi\prec\lambda_\phi$.

 Thus, we have\begin{center} $\frac{n}{2n-2}\geq\frac{1}{2}$      for each $n$ \end{center}
             \begin{center}  $\frac{n+1}{2n-2}\geq\frac{n+3}{2n}$  for $n\leq3$ \end{center}
             \begin{center} $\frac{n+2}{2n-2}\geq\frac{n+4}{2n}$  for $n\leq4$.  \end{center}
              \begin{center} $ \vdots  $ \end{center}
              \begin{center} $\frac{n+k}{2n-2}\geq\frac{n+k+2}{2n}$ for $n\leq k+2$ \end{center}

Clearly, $\lambda_\psi\prec\lambda_\phi$  if and only if $2\leq n\leq 3$, which implies that $|\psi\rangle$ and $|\phi\rangle$ cannot be transformed by LOCC when $n\geq 4$. For simplicity, we say the graph operation cannot simulated by LOCC when $n\geq 4$ since the graph operation is the only variable that discriminate $|\psi\rangle$ from $|\phi\rangle$.\qed
 \end{proof}

 The relationship between graph operation and Von Neumann entropy, LOCC is revealed and a comprehensive understanding on quantum states representation with graphs is provided in this section. Thus the result leads towards a new paradigm of investigating quantum states with graph theory. A natural problem comes subsequently`` What will happen to the Von Neumann entropy by continuing to add an edge on star-like graph and what is the relation between graph operation and LOCC? "

 \section{Star-alike graph, adding an edge on star-like graph}
\label{sec:12}

  In this section, we investigate the specific quantum state corresponding to the star-alike graph(Fig.3 and Fig.4), which means adding an edge on star-like graph. In fact there are two ways for adding an edge on star-like graphs, one is adding adjacent edge on the peripheral vertices while the other is to add nonadjacent edge.
  Similarly, the relationship between the graph operation, i.e., adding an edge on star-like graph,and Von Neumann entropy, LOCC is illustrated. Furthermore, we also present the relationship of Von Neumann entropy and LOCC between quantum states corresponding to the addition of adjacent edge and nonadjacent edge on star-like graph.

  \begin{figure}
  \includegraphics[height=4cm,width=15cm]{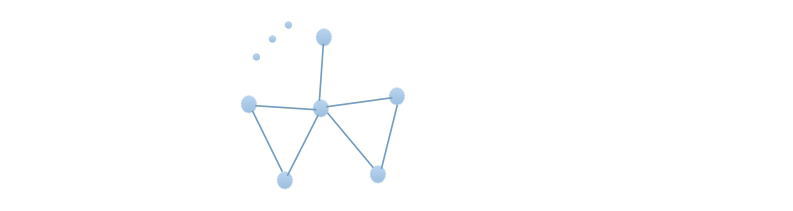}
 \caption{Star-alike1 graph}
\label{fig:1}       
\end{figure}
\begin{figure}
  \includegraphics[height=4cm,width=15cm]{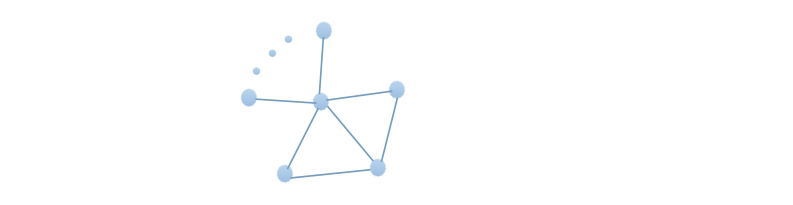}
 \caption{Star-alike2 graph}
\label{fig:1}       
\end{figure}

  \subsection{Spectrum of density operator corresponding to star-alike graph}
  \label{sec:13}
  Similarly, according to the method in section 2.3, we have
 \begin{Theorem}
The spectrum of density operator corresponding to the addition of an adjacent edge  or an nonadjacent edge on the peripheral vertices of star-like graphs is
$\{{\frac{n}{2n+2}^{[1]}, \frac{3}{2n+2}^{[2]}, \frac{1}{2n+2}^{[n-4]}, 0^{[1]}}\}$ or
  $\{{\frac{n}{2n+2}^{[1]}, \frac{4}{2n+2}^{[1]},\frac{2}{2n+2}^{[1]}, \frac{1}{2n+2}^{[n-4]}, 0^{[1]}}\}$.
 \end{Theorem}
\begin{proof}
:To find the spectrum of density operator corresponding to the addition of an adjacent edge on the peripheral vertices of star-like graphs is equal to solve the equation:
 $det(\lambda I-\rho)=0$\\
 That is\\

\begin{eqnarray}
det(\lambda I-\rho)=
\left|
\begin{array}{cccccccc}
\lambda-\frac{n-1}{2n}& \frac{1}{2n} & \ldots&\frac{1}{2n} &\frac{1}{2n}&\frac{1}{2n} &\frac{1}{2n}
\\
\frac{1}{2n}& \lambda-\frac{1}{2n}&\ldots & 0 &0&0&0
\\
\vdots &\vdots & \lambda-\frac{1}{2n}  &\vdots &\vdots&\vdots &\vdots
\\
\frac{1}{2n}&0 &\ldots& \lambda- \frac{2}{2n}& \frac{1}{2n}&0&0
  \\
\frac{1}{2n}&0 &\dots &\frac{1}{2n} & \lambda-\frac{2}{2n}&0&0
\\
\frac{1}{2n}&0 &\ldots&0&0 & \lambda- \frac{2}{2n}& \frac{1}{2n}
  \\
\frac{1}{2n}&0 &\dots &0&0 &\frac{1}{2n} & \lambda-\frac{2}{2n}
\end{array}
\right|=0
\end{eqnarray}\\
 $=(\lambda-\frac{n-1}{2n+2})(\lambda-\frac{1}{2n+2})^{n-5}[(\lambda-\frac{2}{2n+2})^{2}-(\frac{1}{2n+2})^{2}]^{2}$\\
 $-(n-5)(\frac{1}{2n+2})^{2}(\lambda-\frac{1}{2n+2})^{n-6}[(\lambda-\frac{2}{2n+2})^{2}-(\frac{1}{2n+2})^{2}]^{2}$\\
 $-4(\frac{1}{2n+2})^{2}(\lambda-\frac{1}{2n+2})(\lambda-\frac{3}{2n+2})[(\lambda-\frac{2}{2n+2})^{2}-(\frac{1}{2n+2})^{2}]$\\
 $=(\lambda-\frac{n-1}{2n+2})(\lambda-\frac{1}{2n+2})^{n-4}(\lambda-\frac{3}{2n+2})^{2}(\lambda-\frac{1}{2n+2})$\\
 $-(n-1)(\frac{1}{2n+2})^{2}(\lambda-\frac{1}{2n+2})^{n-4}(\lambda-\frac{3}{2n+2})^{2}$\\
 $=[(\lambda-\frac{n-1}{2n+2})(\lambda-\frac{1}{2n+2})-(n-1)\frac{1}{(2n+2)^2}]
    (\lambda-\frac{3}{2n+2})^2(\lambda-\frac{1}{2n+2})^{n-4} $\\
 $=(\lambda^2-\frac{n}{2n+2}\lambda) (\lambda-\frac{3}{2n+2})^2(\lambda-\frac{1}{2n+2})^{n-4}$\\
 $=\lambda(\lambda-\frac{n}{2n+2})(\lambda-\frac{3}{2n})^2(\lambda-\frac{1}{2n+2})^{n-4}$\\
 $=0$\\

 Clearly,  the roots are $\{{\frac{n}{2n+2}^{[1]}, \frac{3}{2n+2}^{[2]}, \frac{1}{2n+2}^{[n-4]}, 0^{[1]}}\}$.

To find the spectrum of density operator corresponding to the addition of a nonadjacent edge on the peripheral vertices of star-like graphs is equal to solve the equation:
$det(\lambda I-\rho)=0$\\

 that is\\

 \begin{eqnarray}
\small {det(\lambda I-\rho)=}
\left|\arraycolsep0pt
\begin{array}{cccccccc}
\lambda-\frac{n-1}{2n+2}& \frac{1}{2n+2} & \ldots&\frac{1}{2n+2} &\frac{1}{2n+2}&\frac{1}{2n+2}
\\
\frac{1}{2n+2}& \lambda-\frac{1}{2n+2}&\ldots & 0 &0&0
\\
\vdots &\vdots & \lambda-\frac{1}{2n+2}  &\vdots &\vdots&\vdots
\\
\frac{1}{2n+2}&0 &\ldots& \lambda- \frac{3}{2n+2}& \frac{1}{2n+2}&\frac{1}{2n+2}
  \\
\frac{1}{2n+2}&0 &\ldots&\frac{1}{2n+2}& \lambda- \frac{2}{2n+2}& 0
  \\
\frac{1}{2n+2}&0 &\dots &\frac{1}{2n+2}&0  & \lambda-\frac{2}{2n+2}
\end{array}
\right|=0
\end{eqnarray}\\
 $=(\lambda-\frac{n-1}{2n+2})(\lambda-\frac{1}{2n+2})^{n-4}
 (\lambda-\frac{2}{2n+2})(\lambda-\frac{4}{2n+2})(\lambda-\frac{1}{2n+2}) $\\
 $-(n-4)(\frac{1}{2n+2})^{2}(\lambda-\frac{1}{2n+2})^{n-5} (\lambda-\frac{2}{2n+2})(\lambda-\frac{4}{2n+2})(\lambda-(\frac{1}{2n+2})$\\
 $-3(\frac{1}{2n+2})^{2}(\lambda-\frac{1}{2n+2})^{n-5} (\lambda-\frac{2}{2n+2})(\lambda-\frac{4}{2n+2})(\lambda-(\frac{1}{2n+2})$\\
 $=(\lambda-\frac{n-1}{2n+2})(\lambda-\frac{1}{2n+2})^{n-4}(\lambda-\frac{2}{2n+2})(\lambda-\frac{4}{2n+2})(\lambda-\frac{1}{2n+2})$\\
 $-(n-1)(\frac{1}{2n+2})^{2}(\lambda-\frac{1}{2n+2})^{n-4}(\lambda-\frac{2}{2n+2})(\lambda-\frac{4}{2n+2})$\\
 $=[(\lambda-\frac{n-1}{2n+2})(\lambda-\frac{1}{2n+2})-(n-1)\frac{1}{(2n+2)^2}]
    (\lambda-\frac{1}{2n+2})^{n-4}(\lambda-\frac{2}{2n+2})(\lambda-\frac{4}{2n+2}) $\\
 $=(\lambda^2-\frac{n}{2n+2}\lambda) (\lambda-\frac{1}{2n+2})^{n-4}(\lambda-\frac{2}{2n+2})(\lambda-\frac{4}{2n+2})$\\
 $=\lambda(\lambda-\frac{n}{2n+2})(\lambda-\frac{1}{2n+2})^{n-4}(\lambda-\frac{2}{2n+2})(\lambda-\frac{4}{2n+2})$\\
 $=0$\\
 Clearly,  the roots are $\{{\frac{n}{2n+2}^{[1]}, \frac{4}{2n+2}^{[1]},\frac{2}{2n+2}^{[1]}, \frac{1}{2n+2}^{[n-4]}, 0^{[1]}}\}$.

 Hence, the spectrum corresponding to star-alike graph is  $\{{\frac{n}{2n+2}^{[1]}, \frac{3}{2n+2}^{[2]}, \frac{1}{2n+2}^{[n-4]}, 0^{[1]}}\}$ or  $\{{\frac{n}{2n+2}^{[1]},\frac{4}{2n+2}^{[1]},\frac{2}{2n+2}^{[1]}, \frac{1}{2n+2}^{[n-4]}, 0^{[1]}}\}$
  \qed
 \end{proof}
 Now that the faithful (and exact) mapping and the specific quantum states connected with star-alike graph is presented, the preparation work for investigating on the relationship between star-alike graph and star-like graph is ready to deploy.
 \subsection{Graph operation and Von Neumann entropy, LOCC}
 \label{sec:14}
   Analogously, in this section, we investigate the relationship  between the graph operation, i.e., adding an edge on star-like graph, and Von Neumann entropy, LOCC.
 \subsubsection{Graph operation and Von Neumann entropy }
 \label{sec:15}
 There are two ways of adding an edge on star-like graph, we have
\begin{Theorem} The corresponding Von Neumann entropy increases under the graph operation adding an edge on star-like graph.\\
   \end{Theorem}
  \begin{proof}
   :It can be easily verified that the spectrum of density operator corresponding star-like graphs is$\{{\frac{1}{2}^{[1]}, \frac{3}{2n}^{[1]}, \frac{1}{2n}^{[n-3]}, 0^{[1]}}\}$, and we have already illustrated the spectrum of density operator corresponding to star-alike graphs $\{{\frac{n}{2n+2}^{[1]}, \frac{3}{2n+2}^{[2]}, \frac{1}{2n+2}^{[n-4]}, 0^{[1]}}\}$ or $\{{\frac{n}{2n+2}^{[1]}, \frac{4}{2n+2}^{[1]},\frac{2}{2n+2}^{[1]}, \frac{1}{2n+2}^{[n-4]}, 0^{[1]}}\}$.

    According to the definition of Von Neumann entropy, the entropy corresponding to quantum states  represented by star graph is presented as follows:\\
 $ S(\rho_{star-like})=-\sum_x\lambda_{x}\log\lambda_{x}$\\
                    $ =-\frac{1}{2}\log(\frac{1}{2})-\frac{3}{2n}\log(\frac{3}{2n})-(n-3)\frac{1}{2n}\log(\frac{1}{2n})$\\
                    $=\frac{1}{2}\log2+\frac{3}{2n}\log(\frac{2n}{3})+\frac{n-3}{2n}\log(2n)$\\
                    $=\frac{1}{2}\log2+\frac{3}{2n}\log(\frac{1}{3})+\frac{1}{2}\log(2n)$\\
    \\ the entropy corresponding to quantum states represented by star-alike graphs is presented as follows:\\

     $ S(\rho_{star-alike})=-\sum_x\lambda_{x}\log\lambda_{x}$\\
       $ =\frac{n}{2n+2}\log(\frac{2n+2}{n})+\frac{6}{2n+2}\log(\frac{2n+2}{3})+\frac{n-4}{2n+2}\log(2n+2)$\\
         $=\log(2n+2)-\frac{n}{2n+2}\log(n)-\frac{6}{2n+2}\log3$\\

     or\\

      $ S(\rho_{star-alike})=-\sum_x\lambda_{x}\log\lambda_{x}$\\
       $ =\frac{n}{2n+2}\log(\frac{2n+2}{n})+\frac{4}{2n+2}\log(\frac{2n+2}{4})+\frac{2}{2n+2}\log(\frac{2n+2}{2})+\frac{n-4}{2n+2}\log(2n+2)$\\
        $=\log(2n+2)-\frac{n}{2n+2}\log(n)-\frac{10}{2n+2}\log2$\\
      \\ Then,\\
     $ S(\rho_{star-alike})-S(\rho_{star-like})=$\\
    $=[\log(2n+2)-\frac{n}{2n+2}\log(n)-\frac{6}{2n+2}\log3]-[\frac{1}{2}\log2+\frac{3}{2n}\log(\frac{1}{3})+\frac{1}{2}\log(2n)]$

     or\\
     $=[\log(2n+2)-\frac{n}{2n+2}\log(n)-\frac{10}{2n+2}\log2]-[\frac{1}{2}\log2+\frac{3}{2n}\log(\frac{1}{3})+\frac{1}{2}\log(2n)]$\\

    $ (S(\rho_{star-alike})-S(\rho_{star-like}))'$\\
    $= [\frac{2}{2n+2}-\frac{1}{2n+2}-\frac{2}{(2n+2)^2}\log(n)+\frac{12}{(2n+2)^2}\log3]-[\frac{3}{2n^2}\log3+\frac{1}{2n}]$\\
    $= \frac{-2}{(2n+2)2n}+\frac{24\log(3) n^2-3(2n+2)^2\log3-4n^2\log(n)}{(2n+2)^22n^2}$, $n\geq 4$\\

     or\\
     $= [\frac{2}{2n+2}-\frac{1}{2n+2}-\frac{2}{(2n+2)^2}\log(n)+\frac{24}{(2n+2)^2}\log2]-[\frac{3}{2n^2}\log3+\frac{1}{2n}]$\\
    $= \frac{-2}{(2n+2)2n}+\frac{48\log(2) n^2-3(2n+2)^2\log3-4n^2\log(n)}{(2n+2)^22n^2}$,$n\geq 5$\\

      Then, $ (S(\rho_{star-alike})-S(\rho_{star-like}))'<0$, therefore the entropy difference is decreasing with each increment in $n$. Also,
     $\lim_{n\rightarrow \infty} [S(\rho_{star-like})-S(\rho_{star})]=0$ and when $n=4$ or $n=5$, $[S(\rho_{star-alike})-S(\rho_{star-like})]>0$,
     hence the entropy difference always remains positive, that is to say the corresponding Von Neumann entropy increases under the graph operation adding an edge on star-like graph. \qed
\end{proof}
%

 \subsubsection{Graph operation and LOCC}
\label{sec:16}
   \begin{Theorem} Suppose $|\psi\rangle$ and $|\phi\rangle$ are pure entangled states shared by Alice and Bob. Let $\rho_{star-like}=tr_B(|\psi\rangle\langle\psi|)$,
    $\rho_{star-alike}=tr_B(|\phi\rangle\langle\phi|)$ corresponding to Alice's system, The quantum states  cannot be transformed to states represented by star-like
    graph by the way of LOCC when $n\geq 5$.\\
\end{Theorem}
   \begin{proof}
   :The density operator and spectrun of the density operator corresponding to star and star-like graphs are given as follows:

   $\rho_{star-like}=tr_B(|\psi\rangle\langle\psi|)$ and
   $\lambda_{\psi}=\{{\frac{1}{2}^{[1]}, \frac{3}{2n}^{[1]}, \frac{1}{2n}^{[n-3]}, 0^{[1]}}\}$

 $\rho_{star-alike}=tr_B(|\phi\rangle\langle\phi|)$  and
 $\lambda_{\phi}=\{{\frac{n}{2n+2}^{[1]}, \frac{3}{2n+2}^{[2]}, \frac{1}{2n+2}^{[n-4]}, 0^{[1]}}\}$\\
 or\\

 $\lambda_{\phi}=\{{\frac{n}{2n+2}^{[1]}, \frac{4}{2n+2}^{[1]},\frac{2}{2n+2}^{[1]}, \frac{1}{2n+2}^{[n-4]}, 0^{[1]}}\}$\\
   With Nielsen theorem, the issue to decide the capability of $|\psi\rangle$ being transformed $|\phi\rangle$ by LOCC is
   equivalent to determine whether $\lambda_\phi\prec\lambda_\psi$.

  Thus, we have\begin{center} $\frac{n}{2n+2}\leq\frac{1}{2}$      for each $n$ \end{center}
             \begin{center}  $\frac{n+3}{2n+2}\leq\frac{n+3}{2n}$  for each $n$ \end{center}
             \begin{center} $\frac{n+6}{2n+2}\leq\frac{n+4}{2n}$  for $n\leq4$.  \end{center}
              \begin{center} $ \vdots  $ \end{center}
              \begin{center} $\frac{n+k+3}{2n+2}\leq\frac{n+k+1}{2n}$ for $n\leq k+1$ \end{center}

Clearly,  $\lambda_\phi\prec\lambda_\psi$ if and only if $n=4$.
 Hence, $|\psi\rangle$ and $|\phi\rangle$ can not be transformed by LOCC when $n\geq 5$, for simplicity, we conclude that the graph operation cannot be simulated by LOCC.\\

              or
              \\
              \begin{center} $\frac{n}{2n+2}\leq\frac{1}{2}$      for each $n$ \end{center}
             \begin{center}  $\frac{n+4}{2n+2}\leq\frac{n+3}{2n}$  for each $n$ \end{center}
             \begin{center} $\frac{n+6}{2n+2}\leq\frac{n+4}{2n}$  for no $n$, since $n\leq4$ and $n\geq5$.  \end{center}
              \begin{center} $ \vdots  $ \end{center}

 Clearly, there is no $n$ satisfying $\lambda_\phi\prec\lambda_\psi$.
 Hence, $|\psi\rangle$ and $|\phi\rangle$ cannot be transformed by LOCC. For simplicity, we conclude that the graph operation cannot be simulated by LOCC when $n\geq 5$.\qed
   \end{proof}
 \subsubsection{Relationship of Von Neumann entropy, LOCC between $Star-alike_1$ and $Star-alike_2$ graph}
 \label{sec:17}
   \begin{Theorem} Suppose $|\phi_1\rangle$ and $|\phi_2\rangle$ are pure entangled states shared by Alice and Bob. Let $\rho_{star-alike_1}=tr_B(|\phi_1\rangle\langle\phi_1|)$,
    $\rho_{star-alike_2}=tr_B(|\phi_2\rangle\langle\phi_2|)$ corresponding to Alice's system, the quantum states represented by one kind of star-alike graph can be transformed into another  by the way of LOCC. Also $S(\rho_{star-alike_1})> S(\rho_{star-alike_2})$. \\
\end{Theorem}
\begin{proof}
   :The density operator and spectrum of the density operator corresponding to star and star-like graphs are

   $\rho_{star-alike_1}=tr_B(|\phi_1\rangle\langle\phi_1|)$ and
  $\lambda_{\phi_1}=\{{\frac{n}{2n+2}^{[1]}, \frac{3}{2n+2}^{[2]}, \frac{1}{2n+2}^{[n-4]}, 0^{[1]}}\}$,\\

 $\rho_{star-alike_2}=tr_B(|\phi_2\rangle\langle\phi_2|)$  and
 $\lambda_{\phi_2}=\{{\frac{n}{2n+2}^{[1]}, \frac{4}{2n+2}^{[1]},\frac{2}{2n+2}^{[1]}, \frac{1}{2n+2}^{[n-4]}, 0^{[1]}}\}$,\\
respectively.\\

   With Nielsen theorem, the issue to decide the capability of $|\phi_1\rangle$ being transformed $|\phi_2\rangle$ by LOCC is
   equivalent to determine whether $\lambda_{\phi1}\prec\lambda_{\phi2}$.

  Thus, we have\begin{center} $\frac{n}{2n+2}\leq\frac{n}{2n+2}$      for each $n$ \end{center}
             \begin{center}  $\frac{n+3}{2n+2}\leq\frac{n+4}{2n+2}$  for each $n$ \end{center}
             \begin{center} $\frac{n+6}{2n+2}\leq\frac{n+6}{2n+2}$   for each $n$ \end{center}
              \begin{center} $ \vdots  $ \end{center}
              \begin{center} $\frac{n+k+2}{2n+2}\leq\frac{n+k+2}{2n+2}$  for each $n$ \end{center}
Clearly, $\lambda_{\phi1}\prec\lambda_{\phi2}$ for each $n$, thus, with \textbf{Lemma 1} in section 2, we have $S(\rho_{star-alike_1})> S(\rho_{star-alike_2})$.
 That is to say $|\phi_1\rangle$ and $|\phi_2\rangle$ can be transformed by LOCC. For simplicity, we say that the graph operation can be simulated by LOCC when $n\geq 5$ and the entropy of adding nonadjacent edge on star-like graph is larger than that of adding adjacent one.\qed
   \end{proof}
 \section{Star-mlike graph, adding $m$ edges on star graph}
\label{sec:18}
 In this section, we investigate the specific quantum state corresponding to the star-mlike graph(Fig.5), which means the addition of nonadjacent edges on peripheral vertices of star graph. Similarly, the relationship between the graph operation, i.e., adding $m$ nonadjacent edges on star graph,and Von Neumann entropy, LOCC is illustrated.

 \begin{figure}
  \includegraphics[height=4cm,width=15cm]{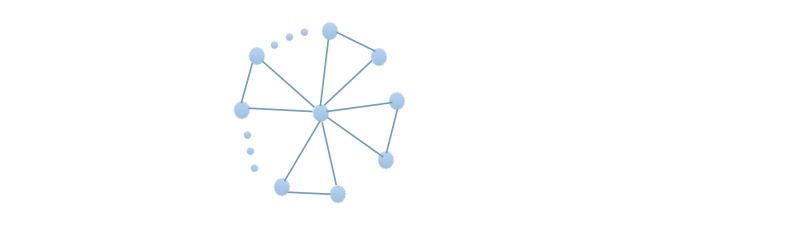}
 \caption{Star-mlike graph}
\label{fig:1}       
\end{figure}
 \subsection{Spectrum of density operator corresponding to adding $m$ nonadjacent edges on star graph}
 \label{sec:19}
  \begin{Theorem}
The spectrum of density operator corresponding to star-mlike graph, which means the addition of $m$ nonadjacent edges on peripheral vertices of star graph, is $\{{\frac{n}{2n+2(m-1)}^{[1]}, \frac{3}{2n+2(m-1)}^{[m]}, \frac{1}{2n+2(m-1)}^{[n-m-2]}, 0^{[1]}}\}$, $m\leq \lfloor \frac{n-1}{2}\rfloor$ .\\
 \end{Theorem}
\begin{proof}

: Finding of the spectrum is equivalent to solve the equation\\
 $det(\lambda I-\rho)=0$. Denote $\alpha=\frac{n-1}{2n+2(m-1)}, \beta=\frac{1}{2n+2(m-1)}, \sigma=\frac{2}{2n+2(m-1)}$
 that is\\
 \begin{eqnarray}
det(\lambda I-\rho)=
\left(
\begin{array}{cccccccc}
\lambda-\alpha & \beta & \ldots&\beta &\beta&\beta &\beta\\
 \beta & \lambda-\beta&\ldots & 0 &0&0&0\\
\vdots &\vdots & \lambda-\beta  &\vdots &\vdots&\vdots &\vdots\\
\beta&0 &\ldots& \lambda- \sigma& \beta&0&0\\
\beta&0 &\dots &\beta & \lambda-\sigma&0&0\\
\beta& \vdots &\vdots & \vdots &\vdots &\vdots&\vdots\\
\beta&0 &\ldots&0&0 & \lambda-\sigma& \beta\\
\beta&0 &\dots &0&0 &\beta & \lambda-\sigma
\end{array}
\right)
\end{eqnarray}\\
$=(\lambda-\frac{n-1}{2n+2(m-1)})(\lambda-\frac{1}{2n+2(m-1)})^{n-2m-1}[(\lambda-\frac{2}{2n+2(m-1)})^{2}-(\frac{1}{2n+2(m-1)})^{2}]^{m}-$\\
 $(n-2m-1)(\frac{1}{2n+2(m-1)})^{2}(\lambda-\frac{1}{2n+2(m-1)})^{n-2m-2}[(\lambda-\frac{2}{2n+2(m-1)})^{2}-(\frac{1}{2n+2(m-1)})^{2}]^{m}$\\
 $-2m(\frac{1}{2n+2(m-1)})^{2}(\lambda-\frac{1}{2n+2(m-1)})^{n-2m}(\lambda-\frac{3}{2n+2(m-1)})[(\lambda-\frac{2}{2n+2(m-1)})^{2}-(\frac{1}{2n+2(m-1)})^{2}]^{m-1}$\\
 $=(\lambda-\frac{n-1}{2n+2(m-1)})(\lambda-\frac{1}{2n+2(m-1)})^{n-m-2}(\lambda-\frac{3}{2n+2(m-1)})^{m}(\lambda-\frac{1}{2n+2(m-1)})$\\
 $-(n-1)(\frac{1}{2n+2(m-1)})^{2}(\lambda-\frac{1}{2n+2(m-1)})^{n-m-2}(\lambda-\frac{3}{2n+2(m-1)})^{m}$\\
 $=[(\lambda-\frac{n-1}{2n+2(m-1)})(\lambda-\frac{1}{2n+2(m-1)})-(n-1)\frac{1}{(2n+2(m-1))^2}]
    (\lambda-\frac{3}{2n+2(m-1)})^m(\lambda-\frac{1}{2n+2(m-1)})^{n-m-2} $\\
 $=(\lambda^2-\frac{n}{2n+2(m-1)}\lambda) (\lambda-\frac{3}{2n+2(m-1)})^m(\lambda-\frac{1}{2n+2(m-1)})^{n-m-2}$\\
 $=\lambda(\lambda-\frac{n}{2n+2(m-1)})(\lambda-\frac{3}{2n+2(m-1)})^m(\lambda-\frac{1}{2n+2(m-1)})^{n-m-2}$\\
 $=0$

Clearly,  the roots are$\{{\frac{n}{2n+2(m-1)}^{[1]}, \frac{3}{2n+2(m-1)}^{[m]}, \frac{1}{2n+2(m-1)}^{[n-m-2]}, 0^{[1]}}\}$. That is to say the spectrum of density operator corresponding to star-mlike graph is$\{{\frac{n}{2n+2(m-1)}^{[1]}, \frac{3}{2n+2(m-1)}^{[m]}, \frac{1}{2n+2(m-1)}^{[n-m-2]}, 0^{[1]}}\}$ . \qed
 \end{proof}

\subsection{Graph operation and Von Neumann entropy, LOCC}
\label{sec:20}
 \subsubsection{Graph operation and Von Neumann entropy }
 \label{sec:21}
  \begin{Theorem}
  The corresponding Von Neumann entropy increases under the graph operation adding an nonadjacent edge on star-mlike graph.\\
   \end{Theorem}
   \begin{proof}
   : The spectrum of density operator corresponding to the addition of $m$ nonadjacent edges on star graphs is $\{{\frac{n}{2n+2}^{[1]}, \frac{3}{2n+2}^{[m]}, \frac{1}{2n+2}^{[n-m-2]}, 0^{[1]}}\}$,
    \begin{enumerate}
  \item By adding one edge on star graph, we get star-like graph. Continuing the addition of a nonadjacent edge on star-like graph, we get star-alike graph.
      The spectrum of star-like graph and star-alike graph is  $\{{\frac{n}{2}^{[1]}, \frac{3}{2n+2}^{[1]}, \frac{1}{2n}^{[n-3]}, 0^{[1]}}\}$  and $\{{\frac{n}{2}^{[1]}, \frac{3}{2n+2}^{[2]}, \frac{1}{2n}^{[n-4]}, 0^{[1]}}\}$, respectively. \textbf{Theorem 5 }already proved that the addition of a nonadjacent edge on star graph increases the entropy.
  \item  Adding $m$ nonadjacent edges on star graph, the resulting spectrum of corresponding density operator is$\{{\frac{n}{2n+2(m-1)}^{[1]}, \frac{3}{2n+2(m-1)}^{[m]}, \frac{1}{2n+2(m-1)}^{[n-m-2]}, 0^{[1]}}\}$. Adding $m+1$ nonadjacent edges on star graph, the resulting spectrum of corresponding density operator is $\{{\frac{n}{2n+2m}^{[1]}, \frac{3}{2n+2m}^{[m+1]}, \frac{1}{2n+2m}^{[n-m-3]}, 0^{[1]}}\}$. The entropy difference between them is

        $S(\rho_{dif})$\\
        $=[\frac{n}{2n+2m}\log(\frac{2n+2m}{n})+\frac{3(m+1)}{2n+2m}\log(\frac{2n+2m}{3})+\frac{n-m-3}{2n+2m}\log(2n+2m)]$\\
        $-[\frac{n}{2n+2(m-1)}\log(\frac{2n+2(m-1)}{n})+\frac{3m}{2n+2(m-1)}\log(\frac{2n+2(m-1)}{3})+\frac{n-m-3}{2n+2(m-1)}\log(2n+2(m-1))]$\\
        $=[\log(2n+2m)-\frac{n}{2n+2m}\log(n)-\frac{3(m+1)}{2n+2m}\log3]-[\log(2n+2(m-1))-\frac{n}{2n+2(m-1)}\log(n)-\frac{3m}{2n+2(m-1)}\log3]$\\

        $S(\rho_{dif})'$
        $=[\frac{2}{2n+2m}-\frac{1}{2n+2m}-\frac{2m}{(2n+2m)^2}\log(n)+\frac{6(m+1)}{(2n+2m)^2}\log3]$\\
        $-[\frac{2}{2n+2(m-1)}-\frac{1}{2n+2(m-1)}-\frac{2(m-1)}{(2n+2(m-1))^2}\log(n)+\frac{6m}{(2n+2m)^2}\log3]$\\
        $=-\frac{2}{(2n+2m)(2n+2(m-1))}-\frac{2m(2n+2(m-1))^2-2(m-1)(2n+2m)^2}{(2n+2m)^2(2n+2(m-1))^2}\log(n)+$\\
        $\frac{6(m+1)(2n+2(m-1))^2-6m(2n+2m)^2}{(2n+2m)^2(2n+2(m-1))^2}\log3$\\
        $=\frac{-2(2n+2m)(2n+2(m-1))+[2(m-1)(2n+2m)^2-2m(2n+2(m-1))^2]}{(2n+2m)^2(2n+2(m-1))^2}\log(n)$\\
        $+\frac{[6(m+1)(2n+2(m-1))^2-6m(2n+2m)^2]}{(2n+2m)^2(2n+2(m-1))^2}\log3$\\
         since $n\geq5$ and $ m\leq \lfloor \frac{n-1}{2}\rfloor$, then $ S(\rho_{dif})'<0$. Hence, the entropy difference decreases with the increment in $n$.
          Also, $\lim_{n\rightarrow \infty}  S(\rho_{dif})=0$ and when $n=5$ $ S(\rho_{dif})>0$, the entropy difference always remains positive, therefore, the corresponding Von Neumann entropy increases under the graph operation by adding a nonadjacent edge on star-mlike graph.
      \item The addition of a nonadjacent edge to star-like graph increases the entropy.  The addition of a nonadjacent edge to star-mlike graph also increases the entropy.
          Thus, the entropy always increases with the addition of a nonadjacent edge on star-relevant graph and the edges number remains $m\leq \lfloor \frac{n-1}{2}\rfloor$.
  \end{enumerate}
 \qed
   \end{proof}
 \subsubsection{  Graph operation and LOCC}
 \label{sec:22}
   \begin{Theorem} Suppose $|\psi\rangle$ and $|\phi\rangle$ are pure entangled states shared by Alice and Bob. Let $\rho_{star_{m}}=tr_B(|\psi\rangle\langle\psi|)$,
    $\rho_{star_{m+1}}=tr_B(|\phi\rangle\langle\phi|)$ corresponding to Alice's system, where $m$ and $m+1$ denote the number of nonadjacent edges addition on star graph. The quantum states $\rho_{star_{m+1}}$ cannot be transformed into states represented by $ \rho_{star_{m}}$
     by the way of LOCC.\\
\end{Theorem}
   \begin{proof}
   :The density operator and spectrum of the density operator  $\rho_{star_{m}}$ and $\rho_{star_{m+1}}$ corresponding to star-relevant graphs are

   $\rho_{star_m}=tr_B(|\psi\rangle\langle\psi|)$ and

   $\lambda_{\psi}=\{{\frac{n}{2n+2(m-1)}^{[1]}, \frac{3}{2n+2(m-1)}^{[m]}, \frac{1}{2n+2(m-1)}^{[n-m-2]}, 0^{[1]}}\}$

 $\rho_{star_{m+1}}=tr_B(|\phi\rangle\langle\phi|)$ and

 $\lambda_{\phi}=\{{\frac{n}{2n+2m}^{[1]}, \frac{3}{2n+2m}^{[m+1]}, \frac{1}{2n+2m}^{[n-m-3]}, 0^{[1]}}\}$,
respectively.\\

   With Nielsen theorem, the issue to decide the capability of $|\phi\rangle$ being transformed $|\psi\rangle$ by LOCC is
   equivalent to determine whether $\lambda_\phi\prec\lambda_\psi$.

  Thus, we have\begin{center} $\frac{n}{2n+2(m-1)}\geq\frac{1}{2n+2m}$      for each $n$ \end{center}
             \begin{center}  $\frac{n+3}{2n+2(m-1)}\geq\frac{n+3}{2n+2m}$  for each $n$  \end{center}
             \begin{center} $ \vdots  $ \end{center}
             \begin{center} $\frac{n+3(m-1)}{2n+2(m-1)}\geq\frac{n+3(m-1)}{2n+2m}$  for each $n$   \end{center}
             \begin{center} $\frac{n+3m+1}{2n+2(m-1)}\geq\frac{n+3m+3}{2n+2m}$  for no $n$, since $n\leq m+3$ and $n\geq m+4$ \end{center}
              \begin{center} $ \vdots  $ \end{center}

Clearly, there is no $n$ to satisfy $\lambda_\psi\prec\lambda_\phi$.
 Therefore, $|\psi\rangle$ and $|\phi\rangle$ cannot be transformed by LOCC. For simplicity, we conclude that the graph operation cannot be simulated by LOCC for each $n$.\qed
   \end{proof}

\section{Discussion and conclusion}
\label{sec:24}
\qquad
  Our analysis of density matrices corresponding to the star-relevant graphs enables us to investigate the quantum state with a new perspective. The relations between 1) graph operation and special quantum operation(LOCC), 2) graph operation and Von Neumann entropy, 3) the specific quantum state and the graph are appropriately investigated in this paper. Furthermore, we investigate the weighted graph defined by Adhikari\textit{ et al}. \cite{17}, where certain conditions should be satisfied for quantum state representation by utilizing graphs. The necessary conditions put forward by Adhikari \textit{et al}. are presented in Theorem 2.3 \cite{17}.

  Combining those conditions and star-relevant graph together, we conclude that the star-relevant graph is not qualified as quantum state representation with the provided weighted graph definition.
   The reason is explained as follows: suppose that the complex weight of star-relevant graph is $e^{iw}$; all star-relevant graphs contain a triangle as a subgraph.
   Thus, we assign the complex weights corresponding to triangle subgraph as $e^{iw_1},e^{iw_2},e^{iw_3}$, separately. Considering the conditions in Theorem 2.3 \cite{17},
   the complex weight on triangle subgraph of star-relevant graph should satisfy
              \begin{equation} e^{iw_1}e^{iw_2}=e^{iw_3} \end{equation}
                \begin{equation}   e^{iw_2}e^{iw_3}=e^{iw_1}\end{equation}
                  \begin{equation}    e^{iw_3}e^{iw_1}=e^{iw_2}\end{equation}
 In other words, the weight corresponding to star-relevant graph should be real and module should be equal to 1.

 The relationship between graph operation and LOCC revealed in this paper is fascinating, i.e., the graph operation cannot be simulated by LOCC with the star-relevant graph.  It encourages us to further investigate more universal relation between the graph operation and LOCC. Unfortunately, this relationship cannot be maintained when we continue adding edges on star-relevant graphs. Here is an example that violates the relation,

\begin{eqnarray}
D(G)=\frac{1}{22}
\left(
\begin{array}{cccccccc}
6&-1&-1&-1&-1&-1&-1&
\\
-1& 2&0&0&0&0&-1&
\\
-1&0& 2&-1&0&0&0&
\\
-1&0&-1&3&-1&0&0&
\\
 -1&0&0&-1&3&-1&0&
 \\
-1&0&0&0&-1&3&-1&
 \\
 -1&-1&0&0&0&-1&3&
\end{array}
\right),
\end{eqnarray}
the corresponding spectrum is $\{0, 0.0576,0.0909 , 0.1364 , 0.1818 ,0.2151 ,0.3182\}$
and adding an edge on the peripheral vertices, we get a ring union star, the corresponding density matrices is presented as follows:
\begin{eqnarray}
D_r(G)=\frac{1}{24}
\left(
\begin{array}{cccccccc}
6&-1&-1&-1&-1&-1&-1&
\\
-1& 3&-1&0&0&0&-1&
\\
-1&-1& 3&-1&0&0&0&
\\
-1&0&-1&3&-1&0&0&
\\
 -1&0&0&-1&3&-1&0&
 \\
-1&0&0&0&-1&3&-1&
 \\
 -1&-1&0&0&0&-1&3&
\end{array}
\right),
\end{eqnarray}
 and the corresponding spectrum is $\{0,0.0833,0.0833,0.1667,0.1667,0.2083,0.2917\}$. It can be verified that it satisfies the majorization relation, thus the graph operation can be simulated by the way of LOCC, which means that the relationship between the graph operation and LOCC isn't universal. Although there is no general rule in the process of adding an edge from a star graph up until a star union a ring graph.  We propose several conjectures through calculation to depict the situation.

\paragraph{Conjectures}
\label{sec:23}
     \begin{enumerate}
     \item The graph operation, adding edges on the peripheral vertices of star graph, cannot be simulated by LOCC when the incremental edge number $m\leq n-3$ ($n$ is the vertex number). Continue adding edge up until a star union a ring graph, the graph operation can be simulated by LOCC.

     \item In the process of adding edge from a star graph up until a star union a ring, quantum states corresponding to those graphs can be transformed by LOCC when the edges number are the same.

     \end{enumerate}

\begin{acknowledgements}
Project supported by NSFC (Grant Nos. 61272514, 61170272, 61121061, 61411146001), NCET (Grant No. NCET-13-0681), the National Development Foundation for Cryptological Research (Grant No. MMJJ201401012) and the Fok Ying Tung Education Foundation (Grant No. 131067).
\end{acknowledgements}




\begin{thebibliography}{}
%
%


\bibitem{1}John Von Neumann.\textit{ Mathematical foundations of quantum mechanics}. No. 2. Princeton university press, 1955. pp.5.
\bibitem{2}Nielsen, Michael A., and Isaac L. Chuang.\textit{Quantum computation and quantum information}. Cambridge university press, 2010. pp.101,573-578.
\bibitem{3}Balachandran, A. P., et al. ``Algebraic approach to entanglement and entropy." \textit{Physical Review A} 88.2 (2013): 022301.
 \bibitem{4}Adhikari, Bibhas, Satyabrata Adhikari, and Subhashish Banerjee. ``Graph representation of quantum states."\textit{ arXiv preprint arXiv}:1205.2747 (2012).
\bibitem{5}Shor, Peter W. ``Polynomial-time algorithms for prime factorization and discrete logarithms on a quantum computer." \textit{SIAM journal on computing }26.5 (1997): 1484-1509.
\bibitem{6}Bennett, Charles H., et al. ``Teleporting an unknown quantum state via dual classical and Einstein-Podolsky-Rosen channels." \textit{Physical review letters }70.13 (1993): 1895.
\bibitem{7} Bouwmeester, Dik, et al. ``Experimental quantum teleportation." \textit{Nature} 390.6660 (1997): 575-579.
 \bibitem{8}Bennett, Charles H., and Stephen J. Wiesner. ``Communication via one-and two-particle operators on Einstein-Podolsky-Rosen states." \textit{Physical review letters} 69.20 (1992): 2881.
 \bibitem{9}Liu, X. S., et al. ``General scheme for superdense coding between multiparties." \textit{Physical Review A} 65.2 (2002): 022304.
 \bibitem{31}Hillery, Mark, Vladim\'{\i}r Bu\v{z}ek, and Andr\'{e} Berthiaume. ``Quantum secret sharing." \textit{Physical Review A} 59.3 (1999): 1829.
 \bibitem{32}Chen, Xiu-Bo, et al. ``Controlled quantum secure direct communication with quantum encryption." \textit{International Journal of Quantum Information} 6.03 (2008): 543-551.
 \bibitem{33}Bennett, Charles H., and Gilles Brassard. ``Quantum cryptography: Public key distribution and coin tossing." \textit{Theoretical Computer Science} 560 (2014): 7-11.

 \bibitem{10}Liu, Bolian, and Hong-Jian Lai.\textit{ Matrices in combinatorics and graph theory}. Vol. 3. Springer Science \& Business Media, 2013.
 \bibitem{11}Brouwer, Andries E., and Willem H. Haemers.\textit{ Spectra of graphs}. Springer Science \& Business Media, 2011.
 \bibitem{12}Newman, Mark. \textit{Networks: an introduction}. Oxford University Press, 2010.
 \bibitem{13}Kocay, William, and Donald L. Kreher.\textit{ Graphs, algorithms, and optimization}. CRC Press, 2004.
 \bibitem{14}Hein, Marc, et al. ``Entanglement in graph states and its applications." \textit{arXiv preprint quant-ph/0602096} (2006).
\bibitem{15}Braunstein, Samuel L., Sibasish Ghosh, and Simone Severini. ``The Laplacian of a graph as a density matrix: a basic combinatorial approach to separability of mixed states." \textit{Annals of Combinatorics} 10.3 (2006): 291-317.

\bibitem{16}Hildebrand, Roland, Stefano Mancini, and Simone Severini. ``Combinatorial laplacians and positivity under partial transpose." \textit{Mathematical Structures in Computer Science} 18.01 (2008): 205-219.
\bibitem{17}Adhikari, Bibhas, Satyabrata Adhikari, and Subhashish Banerjee. ``Graph representation of quantum states." \textit{arXiv preprint arXiv:1205.2747} (2012).
\bibitem{26}Hassan, Ali Saif M., and Pramod S. Joag. ``A combinatorial approach to multipartite quantum systems: basic formulation." \textit{Journal of Physics A: Mathematical and Theoretical} 40.33 (2007): 10251.
   \bibitem{27}Hassan, Ali Saif M., and Pramod S. Joag. ``On the degree conjecture for separability of multipartite quantum states."\textit{ Journal of Mathematical Physics} 49.1 (2008): 012105.
  \bibitem{28}Wu, Chai Wah. ``Conditions for separability in generalized Laplacian matrices and diagonally dominant matrices as density matrices."\textit{ Physics Letters A }351.1 (2006): 18-22.
     \bibitem{29}Wu, Chai Wah. ``On graphs whose Laplacian matrix¡¯s multipartite separability is invariant under graph isomorphism." \textit{Discrete Mathematics} 310.21 (2010): 2811-2814.
       \bibitem{30} Hassan, Ali Saif M. ``Detection and Quantification of Entanglement in Multipartite Quantum Systems Using Weighted Graph and Bloch Representation of States." \textit{arXiv preprint arXiv:0905.0312} (2009).
\bibitem{18}Passerini, Filippo, and Simone Severini. ``The von Neumann entropy of networks." \textit{Available at SSRN 1382662 }(2008).
\bibitem{19}Chitambar, Eric, et al. ``Everything you always wanted to know about LOCC (but were afraid to ask)." \textit{Communications in Mathematical Physics} 328.1 (2014): 303-326.
\bibitem{24}Nielsen, Michael A., and Guifr\'{e} Vidal. ``Majorization and the interconversion of bipartite states." \textit{Quantum Information} \& \textit{Computation} 1.1 (2001): 76-93.
\bibitem{25}Marshall, Albert W., Ingram Olkin, and Barry Arnold. \textit{Inequalities: theory of majorization and its applications}. Springer Science \& Business Media, 2010.
\bibitem{20}Horn, Roger A., and Charles R. Johnson. \textit{Matrix analysis}. Cambridge university press, 2012.pp.387.
\bibitem{21}Watrous, John. ``Theory of quantum information."\textit{ University of Waterloo Fall} (2011).pp.128.

\end{thebibliography}


\end{document}